\documentclass[journal, onecolumn]{IEEEtran}

\usepackage{cite}
\usepackage{graphicx}
\usepackage{eucal}
\graphicspath{{./Images/}}

\usepackage[cmex10]{amsmath}
\usepackage[top=0.7in, bottom=0.7in, left=0.55in, right=0.55in]{geometry}

\usepackage[justification=centering]{caption}
\usepackage{float}
\usepackage{verbatim}
\usepackage{relsize}
\usepackage{amsfonts}
\usepackage{amsthm}
\usepackage{amssymb}
\usepackage{latexsym}
\usepackage{url}
\usepackage{epstopdf}
\usepackage{enumerate}
\usepackage{mathrsfs}
\usepackage{tabulary}
\usepackage{color}
\usepackage{balance}
\newtheorem{theorem}{Theorem}
\newtheorem{lemma}[theorem]{Lemma}
\newtheorem{corollary}[theorem]{Corollary}
\newtheorem{proposition}[theorem]{Proposition}

\theoremstyle{definition}
\newtheorem{definition}[theorem]{Definition}
\newtheorem{example}{Example}
\newtheorem{remark}{Remark}

\newcommand{\bfM}{{\boldsymbol M}}
\newcommand{\bfK}{{\boldsymbol K}}
\newcommand{\bfL}{{\boldsymbol L}}
\newcommand{\bfO}{{\boldsymbol O}}

\newcommand{\bfX}{{\boldsymbol X}}
\newcommand{\bfY}{{\boldsymbol Y}}

\newcommand{\bfu}{{\boldsymbol u}}
\newcommand{\bfv}{{\boldsymbol v}}
\newcommand{\bfw}{{\boldsymbol w}}
\newcommand{\bfx}{{\boldsymbol x}}
\newcommand{\bfy}{{\boldsymbol y}}

\newcommand{\bfell}{{\boldsymbol \ell}}
\newcommand{\bfmu}{{\boldsymbol \mu}}
\newcommand{\one}{{\boldsymbol 1}}

\newcommand{\cC}{{\cal C}}
\newcommand{\cB}{{\cal B}}

\newcommand{\cP}{{\cal P}}
\newcommand{\cS}{{\cal S}}
\newcommand{\cT}{{\cal T}}

\newcommand{\sC}{{\mathcal C}}

\newcommand{\floor}[1]{\left\lfloor #1\right\rfloor}

\newcommand{\etal}{\emph{et al.}}
\newcommand{\ws}{w}

\IEEEoverridecommandlockouts

\begin{document}
\title{Generalized Sphere-Packing Bound for Subblock-Constrained Codes
}

\author{Han Mao Kiah\IEEEauthorrefmark{2}, Anshoo Tandon\IEEEauthorrefmark{1},
         and
        Mehul Motani\IEEEauthorrefmark{1}\\[1mm]
\IEEEauthorblockA{\IEEEauthorrefmark{2} \small School of Physical and Mathematical Sciences, 
Nanyang Technological University, Singapore\\[0mm]}
\IEEEauthorblockA{\IEEEauthorrefmark{1} \small Electrical \& Computer Engineering,
National University of Singapore\\[0mm]}
{Emails: hmkiah@ntu.edu.sg, anshoo.tandon@gmail.com, motani@nus.edu.sg}\vspace{-5mm}}

\maketitle

\begin{abstract}
We apply the generalized sphere-packing bound to two classes of subblock-constrained codes.
{\em \`A la} Fazeli \etal~(2015), we made use of automorphism to significantly reduce the number of variables in the associated linear programming problem. 
{In particular, we study binary \emph{constant subblock-composition codes} (CSCCs), characterized by the property that the number of ones in each subblock is constant, and binary \emph{subblock energy-constrained codes} (SECCs), characterized by the property that the number of ones in each subblock exceeds a certain threshold.}
For CSCCs, we show that the optimization problem is equivalent to finding the minimum of $N$ variables, where $N$ is independent of the number of subblocks.
We then provide closed-form solutions for the generalized sphere-packing bounds for single- and double-error correcting CSCCs.
For SECCs, we provide closed-form solutions for the generalized sphere-packing bounds for single errors in certain special cases. 
{We also obtain improved bounds on the optimal asymptotic rate for CSCCs and SECCs, and provide numerical examples to highlight the improvement.}
\end{abstract}

\section{Introduction}

Subblock-constrained codes are a class of constrained codes where each codeword is divided into smaller subblocks, and each subblock satisfies a certain application-dependent constraint. 
Subblock-constrained codes have recently gained attention as they are suitable candidates for applications such as simultaneous energy and information transfer~\cite{Tandon16_CSCC_TIT}, visible light communication~\cite{Zhao16}, low-cost authentication methods~\cite{Chee14_MCWC}, and powerline communications~\cite{Chee13}. 

In this paper, we discuss two important subclasses of subblock-constrained codes. 
The first subclass are the \emph{constant subblock-composition codes} (CSCCs). Binary CSCCs have varied applications~\cite{Zhao16,Chee14_MCWC,Tandon15_ISIT}, and are characterized by the property that each subblock in every codeword has the same \emph{weight}, i.e. each subblock has the same number of ones. 
The second subclass of subblock-constrained codes that we study are the \emph{subblock energy-constrained codes} (SECCs) which ensure that the energy content in every subblock of each codeword exceeds a certain threshold~\cite{Tandon16_CSCC_TIT}. SECCs have application in simultaneous energy and information transfer~\cite{Tandon16_CSCC_TIT}, and binary SECCs are characterized by the property that the number of ones in each subblock is at least $\ws$. Bounds on the capacity and error exponent for SECCs and CSCCs over noisy channels were presented in~\cite{Tandon16_CSCC_TIT}, while bounds on the SECC and CSCC code size and asymptotic rate, with minimum distance constraint, were analyzed in~\cite{TandonKM18_IT}.

In this paper, we study a modified version of the generalized sphere-packing bound 
{\em \`a la} Fazelli \etal \cite{Fazeli15}.
A specialized version of the generalized bound was first introduced by Kulkarni and Kiyavash \cite{Kulkarni2013} in the context of deletion-correcting code and since then, variants of their method were applied to a myriad of coding problems (see \cite{Fazeli15} for a survey). 
Fazeli {\em et al.} then studied their method in a general setup and provided what is called the \emph{generalized sphere-packing bound}.
Now, the generalized sphere-packing bound is essentially given by the optimal solution of a linear programming problem and 
in most cases, determining its exact value is difficult.
Nevertheless, we apply the symmetry techniques in \cite{Fazeli15} to significantly reduce the size of the linear program and 
our main contributions are the 
\emph{closed-form solutions} of the generalized sphere-packing bound in certain cases for CSCCs and SECCs.

{We also extend the results in~\cite{TandonKM18_IT} to present improved upper bounds on the asymptotic rates for CSCCs and SECCs for a range of relative distance values. These results are obtained by applying a generalized version of the sphere-packing bound (Sec.~\ref{sec:GSP}), and by judiciously choosing appropriate constrained spaces for estimating asymptotic ball sizes.}

\section{Generalized Sphere-Packing Bound} \label{sec:GSP}

We give a modified version of the sphere-packing bound in full generality, 
and then specialize it to the class of codes that we are interested in.

Let $\tau$ be a distance metric defined over $\Sigma^n$ and pick some constrained space $\cS \subseteq \Sigma^n$.
A subset $\sC\subseteq \cS$ is an $(n,d;\cS)$-code if $d=\min \{\tau(\bfx,\bfy) : \bfx,\bfy\in \sC,~\bfx\ne\bfy\}$
and we are interested in determining the value 
$A(n,d;\cS) \triangleq \max \{|\sC| : \mbox{ $\sC$ is an $(n,d;\cS)$-code}\}$.

Fix $d$ and set $t=\floor{(d-1)/2}$. For $\bfx\in \cS$, let $\cB(\bfx,t)$ be the ball $\{\bfy\in\Sigma^n: \tau(\bfx,\bfy)\le t\}$.
We further define $\cT\triangleq\bigcup_{\bfx \in \cS} \cB(\bfx,t)$.
In other words, $\cT$ is the set of all words whose distance is at most $t$ from some word in $\cS$.

We consider a binary matrix $\bfM$ whose rows are indexed by $\cS$ and columns are indexed by $\cT$.
Set 
\[ \bfM_{\bfx,\bfy}=
\begin{cases}
1 & \mbox{if $\tau(\bfx,\bfy)\le t$}, \\ 0 & \mbox{otherwise}.
\end{cases}\]

\begin{theorem}[Fazeli {\em et al.}~\cite{Fazeli15}, Cullina and Kiyavash\cite{CK16}]\label{thm:gsp}
For $d\le n$, set $t$ and $\bfM$ as above.
Then 
\begin{equation}\label{eq:gsp}
\small
A(n,d;\cS) \le \min \left\{\sum_{\bfy\in \cT} Y_\bfy : 
\bfM\bfY\ge \one, Y_\bfy \ge 0 \mbox{ for }\bfy\in \cT\right\}.
\end{equation}

\end{theorem}


Usually, we consider spaces $\cS$ whose size is exponential in $n$.
Therefore, determining the exact value of \eqref{eq:gsp} by solving the linear program directly is computationally prohibitive.
Hence, most authors \cite{Fazeli15,CK16} chose certain feasible points in the linear program \eqref{eq:gsp} and 
evaluated the objective functions to obtain upper bounds on $A(n,d;\cS) $.
In particular, Cullina and Kiyavash \cite{CK16} introduced the local degree iterative algorithm to pick ``good'' feasible points.
In the following theorem, we picked feasible points by  judiciously choosing constrained spaces that contain $\cS$ and estimating the corresponding ball sizes.
This is motivated by Freiman's and Berger's methods \cite{Freiman64A,Berger67} that 
improve the usual sphere-packing bounds for constant weight codes. 

Formally, we choose $\tilde{\cS}\subseteq \Sigma^n$, a subset possibly different from $\cS$.
For $\bfx\in \cS$, define $\cB_{\tilde{\cS}}(\bfx,t)\triangleq \{\bfy\in\tilde{\cS}: \tau(\bfx,\bfy)\le t\}$ and 
set $V^{\min}_{\cS,\tilde{\cS}}(t)\triangleq \min \{|\cB_{\tilde{\cS}}(\bfx,t)|: \bfx\in \cS\}$.

\begin{theorem}\label{thm:ISITA}
Set $t=\floor{(d-1)/2}$. 
For any $\tilde{\cS}\subseteq \Sigma^n$,
if $V^{\min}_{\cS,\tilde{\cS}}(t)\ge 1$, then
\begin{equation}\label{eq:sp-1}
A(n,d;\cS)\le \frac{|\tilde{\cS}|}{V^{\min}_{\cS,\tilde{\cS}}(t)}.
\end{equation}
\end{theorem}
\begin{IEEEproof}
	Abbreviate $V^{\min}_{\cS,\tilde{\cS}}(t)$ with $V$ and consider the vector $\bfY$ with entries $Y_\bfy$ defined as
	\[Y_\bfy  =
	\begin{cases}
	1/V & \mbox{if $\bfy \in \tilde{\cS}$}, \\ 0 & \mbox{otherwise}.
	\end{cases}\]
	We  show that vector $\bfY$ above is a feasible point in the optimization program \eqref{eq:gsp}.
	In other words, we claim that $\bfM\bfY\ge \one$. Indeed, for $\bfx\in \cS$, 
	let $\bfM_\bfx$ denote the row of $\bfM$ that corresponds to $\bfx$. We have that 
	\[ \bfM_\bfx\bfY= |\cB_{\tilde{\cS}\cap \cT}(\bfx,t)|/V= |\cB_{\tilde{\cS}}(\bfx,t)|/V \ge 1,\]
	since $V$ corresponds to the smallest ball volume. To complete the proof, it remains to compute the objective value that is
	$\sum_{\bfy\in \cT} Y_\bfy=\sum_{\bfy\in \tilde{\cS}\cap \cT} Y_\bfy \le |\tilde{\cS}|/V$.
\end{IEEEproof}
Observe that there are exponentially many choices for the constrained space $\tilde{\cS}$.
Nevertheless, in this paper, for the class of CSCCs,  we provide a {\em small family of constrained spaces} and 
show that the {\em optimal solution to \eqref{eq:gsp} must be correspond to one of these constrained spaces}. 
Furthermore, the number of these constrained spaces depends only on $t$ and is {\em independent of the length} $n$.

Another approach to make the linear program \eqref{eq:gsp} tractable exploits symmetries in the linear program.
The approach essentially uses the symmetric structure of the linear program \eqref{eq:gsp}
to define another linear program that has significantly lesser variables and constraints, 
while ensuring the solution to the latter program is the same as the former. 
The method is formally summarised in Theorem \ref{thm:reduction} and 
we remark that similar methods  can be found in linear programming literature (see Margot \cite{Margot:2003}, and B{\H o}di and Herr~\cite{Bodi:2009}). 
Independently, Fazelli {\em et al.} obtained Theorem \ref{thm:reduction} in the specialized setting of finding a fractional transversal in hypergraphs and applied it to various coding problems.
Here, we describe the method in the language of metric spaces.

Recall that $\tau$ is a distance metric defined over $\Sigma^n$.
We say that the permutations $\pi:\Sigma^n\to\Sigma^n$ is an automorphism with respect to $\tau$
if for all $\bfx,\bfy\in \Sigma^n$, we have that $\tau(\bfx,\bfy)=\tau(\pi(\bfx),\pi(\bfy))$.
It is known that the set $G$ of all automorphisms with respect to $\tau$ form a 
subgroup of the symmetric group on the set $\Sigma^n$.

Consider ${\cal X} \subset \Sigma^n$.
A subgroup $H$ of $G$ is {\em $H$-closed} if $\pi(\bfx)\in {\cal X}$ for all $\pi\in H$ and $\bfx\in {\cal X}$.
Recall that $\cT$ is subset of $\Sigma^n$ defined by the constrained space $\cS$ and radius $t$.
Suppose that $\cT$ is $H$-closed and 
let $O_1,O_2,\ldots, O_{N}$ be the orbits under the group action of $H$.

\begin{theorem}[Fazeli \etal \cite{Fazeli15}]\label{thm:reduction}
Given $n$, $d$, $\cS$, we define $t$, $\cT$ and $O_1, O_2,\ldots, O_N$ as above. Then
{
\begin{equation}
A(n,d;\cS) \le \min \left\{\sum_{j\in [N]} Y^*_j : \bfM^*\bfY^*\ge \one, Y^*_j \ge 0 \mbox{ for } j\in [N]\right\}.
\label{eq:reduced-gsp}
\end{equation}
}

\noindent where $\bfM^*$ is a matrix whose rows are indexed by $\cS$, 
columns are indexed by $[N]$, and 
its entries are given by 
\[\bfM^*(\bfx,j)= \sum_{\bfy\in O_j} \bfM(\bfx,\bfy) \mbox{ for all $\bfx\in\cS, \, j\in [N]$}.\]
Note that $\bfM^*(\bfx,j)$ is also the size of the ball $\cB_{O_j}(\bfx,t)$.
\end{theorem}

\subsection{Subblock-Constrained Codes}

In this paper, we focus on the Hamming metric and the following two constrained spaces.
Our constrained spaces are parameterized by integers $m$, $L$ and $\ws$ with $L/2\le \ws\le L$. 
We consider words of $n\triangleq mL$, 
where each word is partitioned into $m$ subblocks, each of length $L$.
The \emph{constant subblock-composition codes} (CSCC) space is the 
 space of all words that have constant weight $\ws$ in each subblock and is denoted by $\cC(m,L,\ws)$.
 On the other hand, the \emph{subblock energy-constrained codes} (SECC) space is the
 space of all words that have weight \emph{at least} $\ws$ in each subblock and is denoted by $\cS(m,L,\ws)$.
The quantities of interest are hence
\begin{align*}
C(m,L,d,\ws) &\triangleq A\left(mL,d;\cC(m,L,\ws)\right),\\
S(m,L,d,\ws) &\triangleq A\left(mL,d;\cS(m,L,\ws)\right).
\end{align*}
Our contributions are as follow:
\begin{enumerate}[(A)]
\item  We provide {\em exact} solutions to the optimization problem given by \eqref{eq:gsp}.
\begin{itemize}
\item  For CSCCs, we show that \eqref{eq:gsp} can computed by finding the minimum amongst a set of at most $N(t)$ values.
We demonstrate that $N(t)$ is independent of $m$ and $L$ and provide a combinatorial interpretation of this value in Section~\ref{sec:auto}. 
Furthermore, each of these $N(t)$ values corresponds to choice of constrained space $\tilde{\cS}$ in Theorem~\ref{thm:ISITA}.
Using this fact, we provide closed-form solutions for the case $t\in\{1,2\}$.
\item For SECCs, we show that \eqref{eq:gsp} can computed by solving a linear program in at most $L^mN(t)$ variables. 
For $t=1$, we provide closed-form solutions for \eqref{eq:gsp} in the following cases: (i) when $m=1$ and $\ws\ge L/2$; (ii) when $\ws=L-1$ and $L\ge m/2$.
\end{itemize}
\item When $d$, or equivalently, $t$, grows linearly with $n$,  the reduced optimization problem remains intractable.
Nevertheless, we provide {\em estimates} to the optimization problem by judiciously choosing appropriate constrained spaces for estimating asymptotic ball sizes and subsquently, applying the generalized sphere-packing bound.
Doing so, we extend the results in~\cite{TandonKM18_IT} to present improved upper bounds on the asymptotic rates for CSCCs and SECCs for a range of relative distance values. 
\end{enumerate}

%

\section{Closed-Form Solutions for \eqref{eq:gsp} for Subblock-Constrained Codes}\label{sec:auto}

For words of length $mL$, we consider the following subgroup of automorphisms 
on $\{0,1\}^{mL}$ with respect to the Hamming metric.
Let $H$ be the set of automorphisms that permute the $m$ subblocks and 
then permute the $L$ coordinates within each block. 
Let $H$ act on the words $\{0,1\}^{mL}$.
Then under this group action, we can index each orbit with an $m$-tuple in
\[\cP(m,L)\triangleq \{[u_1,u_2,\ldots,u_m]: u_1\ge \cdots\ge u_m, 0\le u_i\le L\}.\]
\noindent In other words, if we pick a word in $O_{[u_1,u_2,\ldots,u_m]}$,
$[u_1,u_2,\ldots, u_m]$ is the $m$-tuple obtained by taking the $m$ weights of subblocks
and arranging them in non-increasing order.

\begin{example}Consider $m=2$ and $L=2$.
Under the group action of $H$, the orbits are
\begin{align*}
O_{[0,0]}&=\{0000\},\\
O_{[1,0]}&=\{0001,0010,0100,1000\},\\
O_{[1,1]}&=\{0101, 0110, 1001,1010\},\\
O_{[2,0]}&=\{0011,1100\},\\
O_{[2,1]}&=\{1101,1110,0111,1011\},\\
O_{[2,2]}&=\{1111\}.
\end{align*}
\end{example}
\vspace{2mm}

Notice that 
\begin{align*}
\cC(m,L,\ws) &= O_{[\ws,\ws,\ldots, \ws]},\\ 
\cS(m,L,\ws) &= \bigcup_{u_1\ge u_2\ge \cdots \ge u_m\ge  \ws} O_{[u_1,u_2,\ldots, u_m]}.
\end{align*}
Therefore, both $\cC(m,L,\ws)$ and $\cS(m,L,\ws)$ are $H$-closed.

To apply Theorem~\ref{thm:reduction}, we have
to compute the orbit sizes and determine the number of orbits under the action of $H$.

First, we determine the number of words in the orbit corresponding to some
$m$-tuple $\bfu=[u_1,u_2,\ldots, u_m]$.
To do so, we introduce the following notation for a multinomial.
Let $\bfmu=(\mu_1,\mu_2,\ldots, \mu_\ell)$ be a vector of length $\ell$ such that 
$\sum_{i=1}^\ell \mu_i=m$.
We write
\[ \binom{m}{\bfmu}\triangleq\binom{m}{\mu_1,\mu_2,\ldots, \mu_\ell}=\frac{m!}{\mu_1!\mu_2!\cdots\mu_\ell!}.\]

For $\bfu\in\cP(m,L)$, set 
\[N(\bfu)\triangleq \binom{m}{\bfmu}, \mbox{ where }\mu_j=\#\{i:u_i=j\} \mbox{ for $0\le j\le L$}.\]

Then the size of $O_\bfu$ is given by
\[ |O_\bfu|=N(\bfu)\prod_{i=1}^m \binom{L}{u_i}. \]

To determine the number of orbits, we look at the respective classes of subblock-constrained codes.

\subsection{Constant Subblock-Composition Codes}

First, we consider the space ${\cS}=\cC(m,L,\ws)$ and set $t=\floor{(d-1)/2}$.
Let $\cT$ be the resulting space and our task is to determine the number of orbits that are contained in $\cT$.
Now, the orbits in $\cT$ are indexed by the $m$-tuples in the 
set $\cP(m,L;\ws,t)\triangleq \{\bfu\in\cP(m,L): \sum_{i=1}^m |u_i-\ws|\le t\}$.

Hence, our task is to determine the size of $\cP(m,L;\ws,t)$.
The next proposition states that this number is upper bounded by a value independent of $m$, $L$, and $\ws$.
Its proof is combinatorial in nature and we defer it to Appendix~\ref{app:Nt}.

\begin{proposition}\label{prop:Nt} For all $m, L, \ws, t$, we have $|\cP(m,L;\ws,t)| \le N(t)$, where
\begin{equation}\label{eq:Nt}
N(t)=\sum_{r=0}^t \sum_{i=0}^r p(i)p(r-i),
\end{equation}
\noindent and $p(i)$ is the partition number of $i$. 
Furthermore, we have equality, or $|\cP(m,L;\ws,t)| = N(t)$, whenever
$m\ge t$ and $t\le \min\{\ws,L-\ws\}$.
\end{proposition}

Next, we set $\bfw\triangleq [\ws,\ws,\ldots, \ws]$ and observe that $\cC(m,L,\ws)=O_{\bfw}$. 
Applying Theorem~\ref{thm:reduction}, we reduce \eqref{eq:gsp} to the following optimization program.

\begin{equation}\label{gsp:CSCC}
\min \left\{\sum_{\bfu\in \cP(m,L;\ws,t)} |O_\bfu| Y^*_\bfu : \bfM^*\bfY^*\ge \one, Y_\bfu \ge 0 \right\},
\end{equation}
\noindent $\bfM^*$ is the matrix indexed by $\{\bfw\}\times \cP(m,L;\ws,t)$
whose entries are given by 
\[ \bfM^*_{\bfw,\bfu}\triangleq \big|\, B_{O_{\bfu}}((0^{L-\ws}1^{\ws})^m,t)\, \big|.\]

\begin{example}Consider $m=3$, $L=10$, $\ws=5$ and $t=2$.
Then 
\begin{align*}
\cP(m,L;\ws,t)&=\cP(3,10;5,2)\\
&=\{[5,5,5], [6,5,5], [5,5,4],[7,5,5],[6,6,5],[6,5,4],[5,4,4], [5,5,3]\}.
\end{align*}
Then $\bfM^*$ is given by
\[ \bfM^*=\left(
\begin{array}{cccc cccc}
76 & 15 & 15 & 30 & 75 & 150 & 75 & 30
\end{array}
\right),
\]
and the objective function is given by
{\small
\begin{align*}
&252^3Y_{555} +3(210)(252)^2 Y_{655} +3(210)(252)^2 Y_{554}\\
&+3(120)(252)^2Y_{755} +3(252)(210)^2Y_{665}  +6(252)(210)^2 Y_{654} \\
&+3(252)(210)^2 Y_{544} +3(120)(252)^2Y_{553}
\end{align*}
}\noindent Notice that there are eight variables and the feasible region is described by nine constraints (including the eight nonnegative constraints). 
Hence, each vertex of the feasible region has exactly seven zero components.
Therefore, the optimal value of the linear program is given by
{\small 
\begin{align*}
\min &\Big\{
\frac{252^3}{76},
\frac{3(210)(252)^2}{15} ,
\frac{3(210)(252)^2}{15} ,
\frac{3(120)(252)^2}{30} ,
\frac{3(252)(210)^2}{75}, 
\frac{6(252)(210)^2}{150},
\frac{3(252)(210)^2}{75},
\frac{3(120)(252)^2}{30}
\Big\}\\
&=210565.894.
\end{align*}
}
\end{example}
\vspace{2mm}

More generally, we provide a closed formula for the exact solution to \eqref{gsp:CSCC}, or equivalently \eqref{eq:gsp}.

\begin{proposition}\label{prop:cscc}
For all $m,L,\ws,t$, the exact solution to \eqref{eq:gsp} for CSCCs is given by  
\begin{equation}\label{sol:cscc}
\min\left\{\frac{|O_\bfu|}{\bfM^*_{\bfw,\bfu}}: \bfu\in \cP(m,L;\ws,t) \right\}.
\end{equation}
Furthermore, since $|\cP(m,L;\ws,t)|\le N(t)$, the solution \eqref{sol:cscc} can be computed in time independent of $m$, $L$, and $\ws$.
\end{proposition}

\begin{remark}\hfill
\begin{enumerate}[(i)]
\item For $\bfu\in \cP(m,L;\ws,t)$, if we set the space $\tilde{\cS}$ to be $O_\bfu$ in Theorem~\ref{thm:ISITA}, 
equation \eqref{eq:sp-1} yields the value ${|O_\bfu|}/{\bfM^*_{\bfw,\bfu}}$. 
In other words, in order to obtain the best upper bound for Theorem~\ref{thm:ISITA},
it suffices to consider $\tilde{\cS}$ in the collection $\mathbb{S}\triangleq\{O_\bfu : \bfu\in \cP(m,L;\ws,t)\}$.
\item When $t$ is proportional to $n$, the collection $\mathbb{S}$ is exponential in $n$ and 
hence, it remains computationally prohibitive to check through all possible constrained spaces.
Cullina and Kiyavash~\cite{CK16} introduced the local degree iterative algorithm to 
iteratively improve the objective value of \eqref{eq:gsp} by suitable modifying the current feasible point.
Unfortunately, the algorithm fails to improve {\em all} feasible points that correspond to spaces in $\mathbb{S}$.
\end{enumerate}
\end{remark}
To conclude this subsection, we apply Proposition~\ref{prop:cscc} to the case $t\in \{1,2\}$.

\begin{corollary}\label{cor:cscc}\hfill
\begin{enumerate}[(i)]
\item When $t=1$ and $1\le \ws\le L-1$, we have that
\[
C(m, L, 4, \ws)\le
\begin{cases}
 \frac{\binom{L}{\ws-1}\binom{L}{\ws}^{m-1}}{\ws}, & \mbox{if $\ws\le L/2$},\\
 \frac{\binom{L}{\ws+1}\binom{L}{\ws}^{m-1}}{L-\ws}, & \mbox{otherwise}.
\end{cases}
\]
\item When $t=2$ and $2\le \ws\le L-2$, we have that
\[
C(m, L, 6, \ws)\le
\begin{cases}
 \frac{\binom{L}{\ws-1}^2\binom{L}{\ws}^{m-2}}{\ws^2}, & \mbox{if $\ws\le \frac{L+2}{m+1}$},\\
 \frac{\binom{L}{\ws+1}^2\binom{L}{\ws}^{m-2}}{(L-\ws)^2}, & \mbox{if $\ws\ge\frac{mL-2}{m+1}$},\\
 \frac{\binom{L}{\ws}^{m}}{1+m\ws(L-\ws)}, & \mbox{otherwise}.
\end{cases}
\]
\end{enumerate}
\end{corollary}

\begin{proof}When $t=1$, we have three $m$-tuples in the set $\cP(m,L;\ws,1)$.
We list the $m$-tuples in $\cP(m,L;\ws,1)$ below 
with their corresponding orbit sizes and entries in the matrix $\bfM^*$.

\begin{center}
\renewcommand{\arraystretch}{1.5}
\begin{tabular}{|c|c|c|}\hline
$\bfu$ & $O_\bfu$ & $\bfM^*_{\bfw,\bfu}$\\ \hline
$[\ws, \ws,\ldots, \ws]$ & $\binom{L}{\ws}^m$ &  $1$ \\
$[\ws+1, \ws,\ldots, \ws]$ & $m\binom{L}{\ws+1}\binom{L}{\ws}^{m-1}$ & $m(L-\ws)$ \\
$[\ws, \ws,\ldots, \ws-1]$ & $m\binom{L}{\ws-1}\binom{L}{\ws}^{m-1}$ & $m\ws$ \\ \hline
\end{tabular}
\end{center}

Hence, \eqref{sol:cscc} yields the value 
\[\min\left\{ \frac{\binom{L}{\ws+1}\binom{L}{\ws}^{m-1}}{L-\ws}, \binom{L}{\ws}^m, 
\frac{\binom{L}{\ws-1}\binom{L}{\ws}^{m-1}}{\ws}\right\}.\]
Algebraic manipulations then yield (i).
\vspace{2mm}

For $t=2$, we proceed as before. The eight $m$-tuples in $\cP(m,L;\ws,2)$ below 
with their corresponding orbit sizes and entries in the matrix $\bfM^*$.

\begin{center}
\footnotesize
\renewcommand{\arraystretch}{1.7}
\setlength\tabcolsep{1.5pt}
\begin{tabular}{|c|c|c|}\hline
$\bfu$ & $|O_\bfu|$ & $\bfM^*_{\bfw,\bfu}$\\ \hline
$[\ws, \ws,\ldots, \ws]$ & $\binom{L}{\ws}^m$ &  $1+mL(L-\ws)$ \\
$[\ws+1, \ws,\ldots, \ws]$ & $m\binom{L}{\ws+1}\binom{L}{\ws}^{m-1}$ & $m(L-\ws)$ \\
$[\ws, \ws,\ldots, \ws-1]$ & $m\binom{L}{\ws-1}\binom{L}{\ws}^{m-1}$ & $m\ws$ \\ 
$[\ws+2, \ws,\ldots, \ws]$ & $m\binom{L}{\ws+1}\binom{L}{\ws}^{m-1}$ & $m(L-\ws)$ \\
$[\ws+1, \ws+1,\ldots, \ws]$ & $m\binom{L}{\ws+1}^2\binom{L}{\ws}^{m-2}$ & $\binom{m}{2}(L-\ws)^2$ \\
$[\ws+1, \ws,\ldots, \ws-1]$ & {\scriptsize $m(m-1)\binom{L}{\ws+1}\binom{L}{\ws-1}\binom{L}{\ws}^{m-2}$ }
& {\scriptsize $m(m-1)\ws(L-\ws)$} \\
$[\ws,\ldots, \ws-1, \ws-1]$ & $m\binom{L}{\ws-1}^2\binom{L}{\ws}^{m-2}$ & $\binom{m}{2}\ws^2$ \\
$[\ws, \ws,\ldots, \ws-2]$ & $m\binom{L}{\ws+1}\binom{L}{\ws}^{m-1}$ & $m(L-\ws)$ \\
\hline
\end{tabular}
\end{center}

As before, \eqref{sol:cscc} and algebraic manipulations yield (ii).
\end{proof}

\subsection{Subblock Energy Constrained Codes}

Consider the space ${\cS}=\cS(m,L,\ws)$ and set $t=\floor{(d-1)/2}$.
Let $\cT$ be the resulting space and again, our task is to determine the number of orbits that are contained in $\cT$.
To this end, we set 
\[\cP_{\rm row}(m,L;\ws)\triangleq \{\bfv\in \cP(m,L) : L\ge v_1\ge \cdots\ge v_m\ge \ws\}.\]
Hence, $\cS(m,L,\ws)=\bigcup_{\bfv \in \cP_{\rm row}(m,L;\ws)} O_\bfv$.
As before, we define $\cP(m,L;\bfv,t)= \{\bfu\in\cP(m,L): \sum_{i=1}^m |u_i-v_i|\le t\}$.
Hence, the orbits that partition $\cT$ have indices in the set
\[\cP_{\rm col}(m,L;\ws,t)\triangleq\bigcup_{\bfv\in \cP_{\rm row}(m,L;\ws)}\cP(m,L;\bfv,t). \]

The next proposition states that $N(t)$ (defined in \eqref{eq:Nt}) is an upper bound for $\cP(m,L;\bfv,t)$.
This in turn provides an upper bound for $\cP_{\rm col}(m,L;\ws,t)$, the number of variables in the reduced program.
The proof of Proposition~\ref{prop:Nt-secc} is deferred to Appendix~\ref{app:Nt}.

\begin{proposition}\label{prop:Nt-secc}For all $m, L, \bfv, t$, we have $|\cP(m,L;\bfv,t)| \le N(t)$.
Therefore, $\cP_{\rm col}(m,L;\ws,t)\le L^m N(t)$.
\end{proposition}

Applying Theorem~\ref{thm:reduction}, we reduce \eqref{eq:gsp} to the following optimization program.

\begin{equation}\label{gsp:SECC}
\min \left\{\sum_{\bfu\in \cP_{\rm col}(m,L;\ws,t)} |O_\bfu| Y^*_\bfu : \bfM^*\bfY^*\ge \one, Y_\bfu \ge 0 \right\},
\end{equation}
\noindent $\bfM^*$ is the matrix indexed by $\cP_{\rm row}(m,L;\ws)\times \cP_{\rm col}(m,L;\ws,t)$
whose entries are given by 
\[ \bfM^*_{\bfu,\bfv}\triangleq \big|\, B_{O_{\bfv}}(\bfx,t)\, \big| \mbox{ for some }\bfx\in O_\bfu .\]

\begin{example}\label{exa-secc}
Consider $m=4$, $L=3$, $\ws=2$ and $t=1$. Then
{\footnotesize
\begin{align*}
\cP_{\rm row}(4,3;2) & = \{[3,3,3,3],[3,3,3,2], [3,3,2,2],[3,2,2,2],[2,2,2,2] \} \\ 
\cP_{\rm col}(4,3;2,1) & = \cP_{\rm row}(4,3;2) \cup \{[3,3,3,1], [3,3,2,1], [3,2,2,1], [2,2,2,1]\}.
\end{align*}
}
Then $\bfM^*$ is given by
\[ \bfM^*=
\left(\begin{array}{rrrrrrrrr}
1 & 12 & 0 & 0 & 0 & 0 & 0 & 0 & 0 \\
1 & 1 & 9 & 0 & 0 & 2 & 0 & 0 & 0 \\
0 & 2 & 1 & 6 & 0 & 0 & 4 & 0 & 0 \\
0 & 0 & 3 & 1 & 3 & 0 & 0 & 6 & 0 \\
0 & 0 & 0 & 4 & 1 & 0 & 0 & 0 & 8
\end{array}\right)\, ,\]
and the objective function is given by
\begin{equation*}
Y_{3333} +12 Y_{3332} +54 Y_{3322} + 108Y_{3222} +81Y_{2222}
+12 Y_{3331} + 108 Y_{3321} + 324 Y_{3221} + 324 Y_{2221}.
\end{equation*}
\end{example}

\vspace{2mm}

Since the number of orbits contained in $\cS(m,L;\ws)$ is at most $L^m$, 
the optimization problem \eqref{gsp:SECC} for SECCs
has at most $L^mN(t)$ variables and at most $L^m$ constraints.
Therefore, we have the following proposition.

\begin{proposition}\label{prop:secc}
For all $m,L,\ws,t$, the exact solution to \eqref{eq:gsp} for SECCs 
can be computed in time polynomial in $L^m$ and $N(t)$.
\end{proposition}

Even though we reduce the number of variables from $\Omega(2^{mL})$ to $O(L^mN(t))$, 
the number of variables remains exponential in $m$.
Nevertheless, when $t=1$, we are able to provide closed-form solutions for the optimization problem \eqref{gsp:SECC}.

To solve the linear program, we introduce the notion of optimality certificates.

\begin{definition}\label{def:opt-cert}
A pair $\left(\widetilde{\bfY}^*,\widetilde{\bfX}^*\right)$ is an {\em optimality certificate} for \eqref{gsp:SECC} if the following holds:
\begin{enumerate}[(i)]
\itemsep1em
\item $\widetilde{\bfY}^*$ is feasible solution for the primal problem. In other words, $\bfM^*\widetilde{\bfY}^*\ge \one$ and $\widetilde{Y}^*_\bfu \ge 0$ for all $\bfu\in  \cP_{\rm col}(m,L;\ws,t)$.
\item $\widetilde{\bfX}^*$ is feasible solution for the dual problem. 
In other words, $\widetilde{\bfX}^*\bfM^*\le \bfO$ and $ \widetilde{X}^*_\bfu \ge 0$ for $\bfu\in \cP_{\rm row}(m,L;\ws)$.
Here, $\bfO=(|O_\bfu|)_{\bfu\in  \cP_{\rm col}(m,L;\ws,t)}$.
\item $\sum_{\bfu\in \cP_{\rm col}(m,L;\ws,t)} |O_\bfu| \widetilde{Y}^*_\bfu=
\sum_{\bfu\in \cP_{\rm row}(m,L;\ws)} \widetilde{X}^*_\bfu$.
\end{enumerate}
\end{definition}

Given an optimality certificate, it is then straightforward to obtain the optimal value.

\begin{proposition}[see Chvatal \cite{Chvatal1983}]\label{prop:opt-cert}
If $(\widetilde{\bfY}^*,\widetilde{\bfX}^*)$ is an optimality certificate for \eqref{gsp:SECC}, then
the optimal value for \eqref{gsp:SECC} is given by  $\sum_{\bfu\in \cP(m,L;\ws,t)} |O_\bfu| \widetilde{Y}^*_\bfu$.
\end{proposition}

\begin{example}[Example~\ref{exa-secc} continued]\label{exa-secc-1}
Consider $m=4$, $L=3$, $\ws=2$ and $t=1$. 
Then consider the pair $\left(\widetilde{\bfY}^*,\widetilde{\bfX}^*\right)$, where
\begin{align*}
\widetilde{\bfY}^* & =(0,1/12,1/4,1/4,0,0,0,0,0), \mbox{ and}\\
\widetilde{\bfX}^* & =(1,0,0,18,45/2).
\end{align*}
We verify that 
{\small
\[\bfM^*\widetilde{\bfY}^* \ge \one\mbox{ and }\widetilde{\bfX}^*\bfM^*\le (1, 12, 54, 108, 81,12, 108, 324, 324).\]
}
Also, we check that
\begin{align*}
\sum_{\bfu\in \cP_{\rm col}(m,L;\ws,t)} |O_\bfu| \widetilde{Y}^*_\bfu &=
\frac{12}{12}+\frac{54}{4}+\frac{108}{4}=41.5\\
\sum_{\bfu\in \cP_{\rm row}(m,L;\ws)} \widetilde{X}^*_\bfu &=
1+18+\frac{45}{2}=41.5.
\end{align*}
Hence, $\left(\widetilde{\bfY}^*,\widetilde{\bfX}^*\right)$ satisfies all properties in Definition~\ref{def:opt-cert}. 
Therefore, it follows from Proposition~\ref{prop:opt-cert} that the solution to \eqref{eq:gsp} for SECCs is 41.5.
\end{example}

Therefore, to determine \eqref{gsp:SECC}, it suffices to provide optimality certificates for the problem.
We provide these certificates and the detailed verification in Appendix~\ref{app:opt}.
Here, we state the exact solutions to \eqref{gsp:SECC}.

\begin{proposition}\label{prop:secc1}
Fix $t=1$ and $\ws=L-1$. For all $L$ and $L\ge m/2$, the exact solution to \eqref{gsp:SECC}is as follow.
\begin{enumerate}[(i)]
\item When $m\equiv 0\pmod 4$, the solution is
\[
1+\sum_{i=0}^{\floor{m/4}-1} \frac{1}{4i+4}\left(\binom{m}{4i+2}L^{4i+2}+\binom{m}{4i+3}L^{4i+3}\right);
\]

\item when $m\equiv 1\pmod 4$, the solution is
\[
1+\sum_{i=0}^{\floor{m/4}-1} \frac{1}{4i+5}\left(\binom{m}{4i+3}L^{4i+3}+\binom{m}{4i+4}L^{4i+4}\right);
\]

\item when $m\equiv 2\pmod 4$, the solution is
\[
\sum_{i=0}^{\floor{m/4}} \frac{1}{4i+2}\left(\binom{m}{4i}L^{4i}+\binom{m}{4i+1}L^{4i+1}\right);
\]
\item when $m\equiv 3\pmod 4$, the solution is
\[
\sum_{i=0}^{\floor{m/4}} \frac{1}{4i+3}\left(\binom{m}{4i+1}L^{4i+1}+\binom{m}{4i+2}L^{4i+2}\right).
\]
\end{enumerate}
\end{proposition}

\begin{proposition}\label{prop:secc2}
Fix $t=1$ and $m=1$. For all $L$ and $\ws\ge L/2$, the exact solution to \eqref{gsp:SECC} is as follow.
\begin{enumerate}[(i)]
\item When $L-\ws\equiv 0\pmod 4$, the solution is
\[
1+\sum_{i=0}^{\floor{(L-\ws)/4}-1} \frac{1}{4i+4}\left(\binom{L}{4i+2}+\binom{L}{4i+3}\right);
\]

\item when $L-\ws\equiv 1\pmod 4$, the solution is
\[
1+\sum_{i=0}^{\floor{(L-\ws)/4}} \frac{1}{4i+5}\left(\binom{L}{4i+3}+\binom{L}{4i+4}\right);
\]

\item when $L-\ws\equiv 2\pmod 4$, the solution is
\[
\sum_{i=0}^{\floor{(L-\ws)/4}} \frac{1}{4i+2}\left(\binom{L}{4i}+\binom{L}{4i+1}\right);
\]
\item when $L-\ws\equiv 3\pmod 4$, the solution is
\[
\sum_{i=0}^{\floor{(L-\ws)/4}} \frac{1}{4i+3}\left(\binom{L}{4i+1}+\binom{L}{4i+2}\right).
\]
\end{enumerate}
\end{proposition}

\newpage
\section{Improved bounds on Asymptotic Rates}
In this section, we provide improved upper bounds on the asymptotic rates for CSCCs and SECCs for a range of relative distance values. These results are obtained by judiciously choosing appropriate constrained spaces for estimating asymptotic ball sizes, and by applying the generalized sphere-packing bound. 

\subsection{Constant Subblock-Composition Codes} \label{sec:CSCC_ImprovedSPB}
We present bounds for the CSCC rate in the asymptotic setting where the number of subblocks $m$ tends to infinity, minimum distance $d$ scales linearly with $m$, but $L$ and $w$ are fixed. Formally,  for fixed $0 < \delta < 1$, the asymptotic rate for CSCCs with fixed subblock length $L$, subblock weight parameter $w$, number of subblocks in a codeword $m \to \infty$, and minimum distance $d$ scaling as $d= \lfloor mL\delta \rfloor$ is defined as
\begin{equation}
\gamma(L,\delta,w/L) \triangleq \limsup_{m\to \infty} \frac{\log C(m,L, \lfloor mL\delta \rfloor, w)}{mL} . \label{eq:CSCC_Rate_Def}\\
\end{equation}

The asymptotic CSCC rate, $\gamma(L,\delta,w/L)$, was studied in~\cite{TandonKM18_IT} and it was shown that $\gamma(L,\delta,w/L)=0$ when $\delta \ge \delta^*(w/L)$, where $\delta^*(w/L)$ is defined as
\begin{equation*}
\delta^*(w/L) \triangleq 2\left(\frac{w}{L}\right) \left(1 - \frac{w}{L}\right) .
\end{equation*}
Further, in~\cite{TandonKM18_IT} the following sphere-packing upper bound on $\gamma(L,\delta,w/L)$ was presented.
\begin{theorem}[Tandon \etal{} \cite{TandonKM18_IT}] \label{thm:CSCC_SP_Arxiv}
	For $0 < \delta < \delta^*(w/L)$, we have
	\begin{equation}
	\gamma(L,\delta,w/L) \le \gamma_{SP}(L,\delta,w/L) ,
	\end{equation}
	where $\gamma_{SP}(L,\delta,w/L)$ is defined as
	\begin{align}
	& \frac{1}{L}\log \binom{L}{w} -  \left(\frac{1+\tilde{u}- \lceil\tilde{u}\rceil}{L}\right) \log \binom{w}{\lceil\tilde{u}\rceil} \nonumber \\
	&- \left(\frac{\lceil\tilde{u}\rceil - \tilde{u}}{L}\right) \log \binom{w}{\lfloor\tilde{u}\rfloor} - \left(\frac{\lceil\tilde{u}\rceil - \tilde{u}}{L}\right) \log \binom{L - w}{\lfloor\tilde{u}\rfloor} \nonumber \\
	&- \left(\frac{1 + \tilde{u} - \lceil\tilde{u}\rceil}{L}\right) \log \binom{L - w}{\lceil\tilde{u}\rceil} - \frac{1}{L} h(\lceil\tilde{u}\rceil - \tilde{u}) ,
	\label{eq:CSCC_SP_Arxiv}
	\end{align}
	where $\tilde{u} \triangleq \delta L/4$.
\end{theorem}

We will show that for certain parameters, the above result can be improved by applying the generalized sphere packing formulation in Theorem~\ref{thm:ISITA}. The bound on the asymptotic CSCC rate in Theorem~\ref{thm:CSCC_SP_Arxiv} was obtained by estimating the ball size in the space $\cC(m,L,w)$, and therefore corresponds to the case where $\tilde{\cS} = \cS = \cC(m,L,w)$. In Prop.~\ref{prop:CSCC_SP_codesize_Gen}, we present an upper bound on the optimal CSCC code-size, $C(m,L,d,w)$, by choosing the space $\tilde{\cS} = \cC(m,L,w+1)$.

\begin{proposition} \label{prop:CSCC_SP_codesize_Gen}
	For $2m < d \le 6m$ and $L \ge w+2$, with $t = \floor{(d-1)/2}$ and $\tilde{t} = \floor{(t-m)/2}$, we have
	\begin{equation} \label{eq:CSCC_SP_codesize_Gen}
	C(m,L,d,w) \le \frac{{\binom{L}{w+1}}^m}{ \binom{m}{\tilde{t}} \left[ \binom{L-w}{2} \binom{w}{1}\right]^{\tilde{t}} (L-w)^{m-\tilde{t}} } .
	\end{equation}
\end{proposition}

\begin{IEEEproof}
	We will apply Theorem~\ref{thm:ISITA}, where we choose $\tilde{\cS} = \cC(m,L,w+1)$. Thus
	$|\tilde{\cS}| = \binom{L}{w+1}^m$, and using Theorem~\ref{thm:ISITA}, it suffices to show that
	\begin{equation} \label{eq:CSCC_BallSize_Gen}
	V^{\min}_{\cS,\tilde{\cS}}(t) \ge \binom{m}{\tilde{t}} \left[\binom{L-w}{2} \binom{w}{1}\right]^{\tilde{t}} (L-w)^{m-\tilde{t}} ,
	\end{equation}
	where the constrained CSCC space is $\cS = \cC(m,L,w)$. For $\bfx \in \cS$, let $\Lambda_{\bfx}$ consist of all words $\bfy \in \tilde{\cS}$ which satisfy the following two properties:
	\begin{enumerate}[(i)]
		\item $\tilde{t}$ subblocks of $\bfy$ differ from corresponding subblocks of $\bfx$ in exactly three bit positions.
		\item Remaining $m - \tilde{t}$ subblocks of $\bfy$ differ from corresponding subblocks of $\bfx$ in exactly one bit position.
	\end{enumerate}
	The size of $\Lambda_{\bfx}$ is given by
	\begin{equation*}
	|\Lambda_{\bfx}| = \binom{m}{\tilde{t}} \left[\binom{L-w}{2} \binom{w}{1}\right]^{\tilde{t}} (L-w)^{m-\tilde{t}} .
	\end{equation*}
	For any $\bfy \in \Lambda_{\bfx}$, we observe that $\tau(\bfx,\bfy) = 3\tilde{t} + (m-\tilde{t}) \le t$, and thus $\Lambda_{\bfx} \subseteq \cB_{\tilde{\cS}}(\bfx,t)$. 
	Finally, the inequality in \eqref{eq:CSCC_BallSize_Gen} follows because $\cB_{\tilde{\cS}}(\bfx,t) \ge |\Lambda_{\bfx}|$ for all $\bfx \in \cS$.
\end{IEEEproof}

The following theorem applies Prop.~\ref{prop:CSCC_SP_codesize_Gen} to provide an upper bound on the asymptotic rate for CSCCs.
\begin{theorem}
	For $2/L < \delta < 6/L \le \delta^*(w/L)$, we have
	\begin{equation}
	\gamma(L,\delta,w/L) \le \acute{\gamma}_{SP}(L,\delta,w/L) , 
	\end{equation}
	where $\acute{\gamma}_{SP}(L,\delta,w/L)$ is defined as
	\begin{align} \label{eq:CSCC_SP_new}
	\frac{1}{L} \log&\binom{L}{w+1} -\left(\frac{\delta}{4} - \frac{1}{2L}\right) \log\left[\binom{L-w}{2} \binom{w}{1}\right] \nonumber \\
	&-\frac{1}{L} h\left( \frac{L\delta}{4} - \frac{1}{2}\right) - \left(\frac{3}{2L} - \frac{\delta}{4}\right) \log(L-w) .
	\end{align}
\end{theorem}
\begin{IEEEproof}
	We will combine \eqref{eq:CSCC_Rate_Def} and \eqref{eq:CSCC_SP_codesize_Gen} to prove the theorem. Towards this, note that when $d$ scales as $d = \floor{mL\delta}$, and $\tilde{t} = \floor{(t-m)/2}$ with $t = \floor{(d-1)/2}$, then we have
	\begin{align}
	\limsup_{m \to \infty} \frac{1}{mL} \log \left[\binom{m}{\tilde{t}}\right] &= \frac{1}{L} h\left( \frac{L\delta}{4} - \frac{1}{2}\right) , \label{eq:CSCC_GenSP_step1} \\
	\limsup_{m \to \infty} \frac{\tilde{t}}{mL} &= \left(\frac{\delta}{4} - \frac{1}{2L}\right) , \label{eq:CSCC_GenSP_step2} \\
	\limsup_{m \to \infty} \frac{m-\tilde{t}}{mL} &= \left(\frac{3}{2L} - \frac{\delta}{4} \right) . \label{eq:CSCC_GenSP_step3}
	\end{align}
	The proof is now complete by combining \eqref{eq:CSCC_GenSP_step1}, \eqref{eq:CSCC_GenSP_step2}, \eqref{eq:CSCC_GenSP_step3}, with \eqref{eq:CSCC_Rate_Def} and \eqref{eq:CSCC_SP_codesize_Gen}.
\end{IEEEproof}

\begin{figure}[t]
	\centering
	\includegraphics[width=0.45\textwidth]{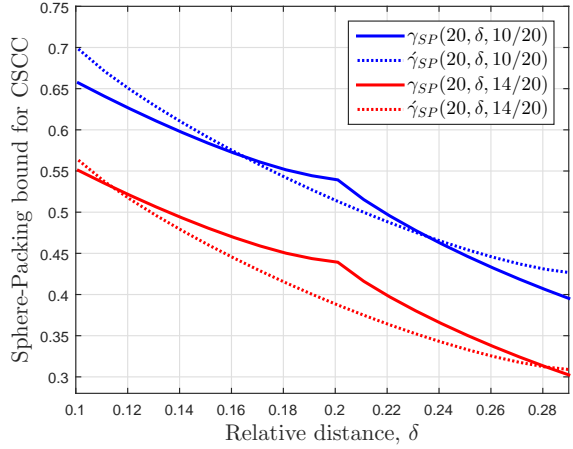}
	\caption{Comparison of sphere-packing upper bounds for the CSCC asymptotic rate $\gamma(L,\delta,w/L)$ for $L=20$.}
	\label{Fig:CSCC_GenSP}
\end{figure}

\begin{proposition} \label{prop:CSCC_SP_new_better}
	For $L/2 \le w < L-1$ and $\delta = 4/L$, we have $$\acute{\gamma}_{SP}(L,\delta,w/L) < \gamma_{SP}(L,\delta,w/L)$$
\end{proposition}
\begin{IEEEproof}
	When $\delta = 4/L$, using \eqref{eq:CSCC_SP_Arxiv} we get
	\begin{equation} \label{eq:CSCC_SP_Arxiv_delta_fix}
	\gamma_{SP}(L,4/L,w/L) = \frac{1}{L} \log \binom{L}{w} - \frac{1}{L} \log \left((L-w)w\right).
	\end{equation}
	On the other hand, using \eqref{eq:CSCC_SP_new} we observe that $\acute{\gamma}_{SP}(L,4/L,w/L)$ is equal to
	\begin{align}
	\frac{1}{L} &\log \binom{L}{w+1} - \frac{1}{2L} \log \left(2 (L-w)^2 (L-w-1) w \right) \nonumber \\
	=\frac{1}{L} &\log \binom{L}{w} - \frac{1}{2L} \log \left(2 (w+1)^2 (L-w-1) w \right). \label{eq:CSCC_SP_new_delta_fix}
	\end{align}
	The proposition is now proved by comparing \eqref{eq:CSCC_SP_Arxiv_delta_fix} and \eqref{eq:CSCC_SP_new_delta_fix}, and observing that $2w (L-w-1) \ge (L-w)^2$ when $L/2 \le w < L-1$.
\end{IEEEproof}
\emph{Remark}: As $\acute{\gamma}_{SP}(L,\delta,w)$ and $\gamma_{SP}(L,\delta,w)$ are both continuous functions of $\delta$, we observe that Prop.~\ref{prop:CSCC_SP_new_better} implies that for a certain interval around $\delta = 4/L$, the upper bound on the CSCC asymptotic rate given by $\acute{\gamma}_{SP}(L,\delta,w)$ is an improved upper bound on the CSCC rate compared to $\gamma_{SP}(L,\delta,w)$. This is depicted in Fig.~\ref{Fig:CSCC_GenSP} for the case where $L=20$ and $w \in \{10,14\}$. Fig.~\ref{Fig:CSCC_GenSP} shows that $\acute{\gamma}_{SP}(L,\delta,w) < \gamma_{SP}(L,\delta,w)$ for a range of $\delta$ values around $\delta = 4/L = 0.2$.

\subsection{Subblock Energy-Constrained Codes} \label{sec:SECC_ImprovedSPB}
We provide an upper bound on the asymptotic SECC rate when the number of subblocks $m$ tends to infinity, minimum distance $d$ scales linearly with $m$, and parameters $L$, $w$ are fixed. Formally,  for fixed $0 < \delta < 1$, the asymptotic rate for SECCs is defined as
\begin{equation} \label{eq:SECC_RateDef1}
\sigma(L,\delta,w/L) \triangleq \limsup_{m\to \infty} \frac{\log S\left( m,L,\floor{mL\delta},w\right)}{mL}. 
\end{equation}
Further, we introduce the notation $\binom{L}{\ge w}$ which we define as
\begin{equation*}
\binom{L}{\ge w} \triangleq \sum_{j=w}^L\binom{L}{j} .
\end{equation*}

First, we apply Theorem~\ref{thm:ISITA} to present an upper bound on the optimal code size for SECCs.
\begin{proposition} \label{prop:SECC_SP_optimize_m0}
	For $d \le 2m+1$, $0\le m_0\le m$, and $t = \floor{(d-1)/2}$, we have
	\begin{equation} \label{eq:SECC_SP_optimize_m0}
	S(m,L,d,w) \le \frac{\displaystyle \binom{L}{\ge w-1}^{m_0} \binom{L}{\ge w}^{m-m_0}}
	{\displaystyle \sum_{\substack{t_1, t_2 \\ t_1 + t_2 \le t}} \binom{m_0}{t_1} \binom{m-m_0}{t_2} L^{t_1} (L-w)^{t_2}} .
	\end{equation}
\end{proposition}
\begin{IEEEproof}
	We will apply Theorem~\ref{thm:ISITA}, and choose $\tilde{\cS} \subset \{0,1\}^{mL}$ to be the space
	where the first $m_0$ subblocks have weight at least $w-1$, and the remaining $m-m_0$ subblocks
	have weight at least $w$, with fixed subblock length $L$. Thus
	$|\tilde{\cS}| = \binom{L}{\ge w-1}^{m_0} \binom{L}{\ge w}^{m-m_0}$,
	and using Theorem~\ref{thm:ISITA}, it suffices to show that
	\begin{equation} \label{eq:SECC_BallSize_m0}
	V^{\min}_{\cS,\tilde{\cS}}(t) \ge \sum_{\substack{t_1, t_2 \\ t_1 + t_2 \le t}} \binom{m_0}{t_1} \binom{m-m_0}{t_2} L^{t_1} (L-w)^{t_2} .
	\end{equation}
	For $\bfx \in \cS(m,L,w)$, let $\bfx_{[i]}$ denote the $i$th subblock of $\bfx$, and hence
	$\bfx = (\bfx_{[1]} \bfx_{[2]} \ldots \bfx_{[m]})$. Let $\Lambda_{\bfx}$ be defined as
	\begin{equation*}
	\Lambda_{\bfx} \triangleq \{\bfy \in \tilde{\cS} : \tau(\bfx,\bfy) \le t, \tau(\bfx_{[i]},\bfy_{[i]}) \le 1, i \in \{1,\ldots,m\} \}.
	\end{equation*}
	Let $\bfy \in \tilde{\cS}$ be such that $t_1$ (resp. $t_2$) subblocks out of the first $m_0$ (resp. last $m-m_0$) subblocks of $\bfy$ differ in exactly one bit from corresponding subblocks of $\bfx$, with $t_1 + t_2 \le t$. Then $\bfy \in \Lambda_{\bfx}$, and 
	\begin{equation*}
	|\Lambda_{\bfx}| \ge \sum_{\substack{t_1, t_2 \\ t_1 + t_2 \le t}} \binom{m_0}{t_1} \binom{m-m_0}{t_2} L^{t_1} (L-w)^{t_2} .
	\end{equation*}
	Note that the inequality above holds for every $\bfx \in \cS(m,L,w)$. Finally, the inequality in \eqref{eq:SECC_BallSize_m0} follows because $\Lambda_{\bfx} \subseteq \cB_{\tilde{\cS}}(\bfx,t)$ for every $\bfx \in \cS(m,L,w)$.
\end{IEEEproof}

The following theorem gives an upper bound on the SECC rate $\sigma(L,\delta,w/L)$.
\begin{theorem} \label{thm:SECC_Gen_SP}
	For $0 < \delta < 2/L$, we have
	\begin{equation}
	\sigma(L,\delta,w/L) \le R_1 - \hat{\alpha} \nu , \label{eq:SECC_Gen_SP}
	\end{equation}
	where
	\begin{align}
	R_1 &\triangleq \frac{\log \binom{L}{\ge w}}{L} - \frac{h(\delta L/2)}{L} - \frac{\delta}{2} \log (L-w) \label{eq:R1_def}\\
	h(x) &\triangleq -x \log(x) - (1-x) \log(1-x) \nonumber \\
	\nu &\triangleq \frac{\delta}{2} \log\left(\frac{L}{L-w}\right) - 
	\frac{1}{L} \log\left[\frac{\binom{L}{\ge w -1}}{\binom{L}{\ge w}}\right] \label{eq:R2_def} \\
	\hat{\alpha} &\triangleq \begin{cases}
	0, \mathrm{~if~} \nu \le 0 \\
	1, \mathrm{~if~} \nu > 0
	\end{cases} \label{eq:alpha_hat_def}
	\end{align}
\end{theorem}
\begin{IEEEproof}
	For $0 \le m_0 \le m$, let $\alpha = m_0/m$ with $t_1 = \floor{t \alpha}$ and $t_2 = \floor{t (1-\alpha)}$. Then $t_1 + t_2 \le t$, and it follows from \eqref{eq:SECC_SP_optimize_m0} that 
	{\small{\begin{equation} \label{eq:SECC_SP_optimize_m0_v2}
			S(m,L,d,w) \le \frac{\displaystyle \binom{L}{\ge w-1}^{m_0} \binom{L}{\ge w}^{m-m_0}}
			{\displaystyle  \binom{m_0}{\floor{t \alpha}} \binom{m-m_0}{\floor{t (1-\alpha)}} L^{\floor{t \alpha}} (L-w)^{\floor{t (1-\alpha)}} } ,
			\end{equation}}}
	Combining~\eqref{eq:SECC_RateDef1} and \eqref{eq:SECC_SP_optimize_m0_v2}, we get
	\small
	\begin{align*}
	\sigma(L,\delta,w/L) &\le \frac{\alpha}{L}\log \binom{L}{\ge w-1} + \frac{1-\alpha}{L}\log \binom{L}{\ge w} \\
	&- \frac{\alpha}{L} h\left(\frac{\delta L}{2}\right) - \frac{(1-\alpha)}{L} h\left(\frac{\delta L}{2}\right) \\
	&- \frac{\alpha \delta}{2} \log(L) - \frac{(1-\alpha) \delta}{2} \log (L-w) .
	\end{align*}
	\normalsize
	By combining the coefficients of $\alpha$, the above inequality can be expressed as
	\begin{equation} \label{eq:SECC_GenSP_Bound1}
	\sigma(L,\delta,w/L) \le R_1 - \alpha \nu ,
	\end{equation}
	where $R_1$ and $\nu$ are given by \eqref{eq:R1_def} and \eqref{eq:R2_def}, respectively. 
	The above bound on SECC rate holds for all $\alpha \in [0,1]$,
	and hence the right side in~\eqref{eq:SECC_GenSP_Bound1} is minimized by choosing $\alpha = \hat{\alpha}$,
	with $\hat{\alpha}$ given by~\eqref{eq:alpha_hat_def}.
\end{IEEEproof}

We observe that the upper bound on $\sigma(L,\delta,w/L)$, given by Theorem~\ref{thm:SECC_Gen_SP}, can equivalently be expressed as
\begin{equation*}
\sigma(L,\delta,w/L) \le \min\{R_1, R_1-\nu\} ,
\end{equation*}
where $R_1$ corresponds to the sphere packing bound on the asymptotic rate $\sigma(L,\delta,w/L)$ when space $\tilde{\cS}$ is chosen to be $\tilde{\cS} = \cS(m,L,w)$, while $R_1-\nu$ corresponds to the sphere packing bound when space $\tilde{\cS}$ is chosen to be $\tilde{\cS} = \cS(m,L,w-1)$.
\begin{corollary} \label{cor:SECC_SP_new_bound_better}
	$R_1-\nu$, the upper bound on SECC rate obtained by choosing $\tilde{\cS} = \cS(m,L,w-1)$, is less than $R_1$, the upper bound on SECC rate obtained by choosing $\tilde{\cS} = \cS(m,L,w)$, for the following range of $\delta$ values
	\begin{equation}
	\frac{2}{L \log [L/(L-w)]} \log\left[\frac{\binom{L}{\ge w -1}}{\binom{L}{\ge w}}\right] < \delta < \frac{2}{L} . \label{eq:SECC_delta_range_new_better}
	\end{equation}	
\end{corollary}
\begin{IEEEproof}
	Follows from  \eqref{eq:R2_def}.
\end{IEEEproof}

An alternate sphere-packing bound was presented in~\cite{TandonKM18_IT}, where it was shown that $\sigma(L,\delta,w/L)$, the asymptotic rate for SECCs, is upper bounded by 
\begin{equation}
\sigma_{SP} \triangleq \frac{\log \binom{L}{\ge w}}{L} - \frac{h(\delta L/4)}{L}  - \frac{\delta}{4} \log \left((L-w)(w+1)\right) . \label{eq:SECC_SP_rate_old}
\end{equation}

\begin{figure}[t]
	\centering
	\includegraphics[width=0.45\textwidth]{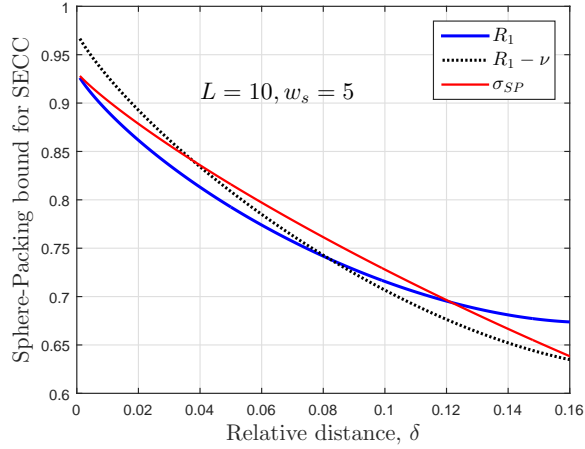}
	\caption{Comparison of sphere-packing upper bounds for the SECC asymptotic rate $\sigma(L,\delta,w/L)$.}
	\label{Fig:SECC_GenSP}
\end{figure}

Fig.~\ref{Fig:SECC_GenSP} compares different sphere-packing bounds for the SECC asymptotic rate $\sigma(L,\delta,w/L)$ as a function of $\delta$ with fixed $L=10$, and $w=5$. As shows in Cor.~\ref{cor:SECC_SP_new_bound_better}, it is observed in Fig.~\ref{Fig:SECC_GenSP} that the upper bound given by $R_1-\nu$ is less than $R_1$ for $\delta > \frac{2}{L \log [L/(L-w)]} \log\left[\frac{\binom{L}{\ge w -1}}{\binom{L}{\ge w}}\right] = 0.0821$.

\newpage

\section{Concluding Remarks}
We study the generalized sphere-packing bound for two classes of subblock-constrained codes, 
namely, CSCCs and SECCs. 
Using automorphisms, we significantly reduce the number of variables in 
the associated linear programming problem.

For CSCCs, to determine the upper bound for $C(m,L,d,w)$, 
we show that the generalised sphere-packing bound can be obtained by finding the minimum of $N(t)$ values, where $t=\floor{(d-1)/2}$ and $N(t)$ is independent of $m$, $L$ and $w$.
We then provide closed-form solutions for the generalized sphere-packing bounds for single- and double-error correcting CSCCs in Corollary~\ref{cor:cscc}.

In contrast, for SECCs, the generalised sphere-packing bound for $S(m,L,d,w)$ is obtained via a linear program involving at most $L^mN(t)$ variables.
Nevertheless, in the special cases, we solved the linear program and closed-form solutions are provided in Propositions~\ref{prop:secc1} and \ref{prop:secc2}.

{Further, we extended the results in~\cite{TandonKM18_IT} to present improved upper bounds on the asymptotic rates for both CSCCs and SECCs for a range of relative distance values.}

\appendices

\section{Upper Bound on the Number of Orbits}\label{app:Nt}

In this appendix, we prove Propositions~\ref{prop:Nt} and~\ref{prop:Nt-secc}.

Recall that $\cP(m,L;\ws,t)\triangleq \{\bfu\in\cP(m,L): \sum_{i=1}^m |u_i-\ws|\le t\}$
and our task is to provide an upper bound on the size of $\cP(m,L;\ws,t)$.

To this end, we consider the notion of partitions and partitions into parts of two kinds.
Specifically, for a fixed value of $t$, we say that a tuple $\lambda=(\ell_1,\ell_2,\ldots, \ell_s)$ is a {\em partition} of $t$ 
if $\ell_1\ge \ell_2\ge \cdots \ge \ell_s>0$ and $\sum_{i=1}^s\ell_i=t$. 
A pair $(\lambda_1,\lambda_2)$ is a {\em partition of $t$ into parts of two kinds} if 
$\lambda_i$ is a partition of $t_i$ for $i\in\{1,2\}$ and $t_1+t_2=t$.
We then set $\Lambda(t)$ to be the collection of all partitions of $t$ into parts of two kinds.
The size of $\Lambda(t)$ is (see for example, sequence A000712 -- \url{https://oeis.org/A000712}) as follow.
\begin{equation}\label{eq:Lambda}
|\Lambda(t)|=\sum_{i=0}^t p(i)p(t-i),
\end{equation}
\noindent where  $p(i)$ is the partition number of $i$. 

To give an upper bound on the size of $\cP(m,L;\ws,t)$, 
we define the map $\phi:\cP(m,L;\ws,t)\to \bigcup_{r=0}^t \Lambda(r)$.
For $\bfu\in \cP(m,L;\ws,t)$, set $(\ell_1,\ell_2,\ldots, \ell_m)=\bfu-\bfw$
where $\bfw=[w,w,\ldots, w]$. Then we find $i$ and $j$ such that 
$\ell_{i-1}>0$, $\ell_{j+1}<0$ and $\ell_{i}=\ell_{i+1}=\cdots=\ell_j=0$.
Finally, we let $\lambda_1=(\ell_1,\ell_2,\ldots, \ell_{i-1})$ and $\lambda_2=(-\ell_m,-\ell_{m-1},\ldots, -\ell_{j+1})$,
and set $\phi(\bfu)=(\lambda_1,\lambda_2)$.

\begin{lemma}\label{lem:app}
Let $\phi:\cP(m,L;\ws,t)\to \bigcup_{r=0}^t \Lambda(r)$ be defined as above.
Then $\phi$ is a well-defined injective map.
Furthermore, if $m\ge t$, we have that $\phi$ is a bijection.
\end{lemma}

\begin{proof}Let $\lambda_1$ and $\lambda_2$ be defined as above.
Since $u_1\ge u_2\ge \cdots\ge u_m$, we have that $\lambda_1$ and $\lambda_2$ are partitions.
Since $\sum_{i=1}^m |u_i-\ws|=\sum_{i=1}^m \ell_i=r$, 
we have that $(\lambda_1,\lambda_2)$ is a partition of $r$ into parts of two kinds. 
As $r\le t$, we have that $(\lambda_1,\lambda_2)$ belongs to $\bigcup_{r=0}^t \Lambda(r)$ and 
the map $\phi$ is therefore well-defined. 

To demonstrate injectivity, suppose that $(\lambda_1,\lambda_2)$ is the image $\phi(\bfu)$ and 
we find $\bfu$ from $(\lambda_1,\lambda_2)$.
First, we pad $\lambda_1$ with $m-|\lambda_1|$ zeroes to obtain $\bfell_1$.
Next, we reverse $\lambda_2$ and pad the result $m-|\lambda_2|$ zeroes at the front to obtain $\bfell_2$.
It is then easy to verify that $\bfu=\bfell_1-\bfell_2+\bfw$.

Finally, we assume that $m\ge t$ and we consider $(\lambda_1,\lambda_2)\in \bigcup_{r=0}^t \Lambda(r)$.
Since $\lambda_1$ and $\lambda_2$ are partitions of $t_1$ and $t_2$ with $t_1+t_2\le t$, 
the sum of their lengths $|\lambda_1|+|\lambda_2|$ is at most $t$.
Define $\bfell_1$ and $\bfell_2$ as in the preceding paragraph and 
we have that the support sets of $\bfell_1$ and $\bfell_2$ are disjoint.
Therefore, the tuple $\bfu=\bfell_1-\bfell_2+\bfw$ has non-increasing entries and belong to $\cP(m,L;\ws,t)$.
Since $\phi(\bfu)=(\lambda_1,\lambda_2)$, the map $\phi$ is surjective and therefore, bijective.
\end{proof}

\eqref{eq:Nt} then follows from the fact that $\phi$ is injective and the formula \eqref{eq:Lambda}.
To complete the proof of Proposition~\ref{prop:Nt}, we observe that $\phi$ is a bijection whenever $m\ge t$.

\vspace{2mm}
We next prove Propostion~\ref{prop:Nt-secc}.
Recall that $\cP(m,L;\bfv,t)\triangleq \{\bfu\in\cP(m,L): \sum_{i=1}^m |u_i-v_i|\le t\}$
and our task is to provide an upper bound on the size of $\cP(m,L;\bfv,t)$.

As before, 
we define the map $\phi_\bfv:\cP(m,L;\ws,t)\to \bigcup_{r=0}^t \Lambda(r)$.
For $\bfu\in \cP(m,L;\ws,t)$, set $(\ell_1,\ell_2,\ldots, \ell_m)=\bfu-\bfv$.
Then we find $i$ and $j$ such that 
$\ell_{i-1}>0$, $\ell_{j+1}<0$ and $\ell_{i}=\ell_{i+1}=\cdots=\ell_j=0$.
Finally, we let $\lambda_1=(\ell_1,\ell_2,\ldots, \ell_{i-1})$ and $\lambda_2=(-\ell_m,-\ell_{m-1},\ldots, -\ell_{j+1})$,
and set $\phi_\bfv(\bfu)=(\lambda_1,\lambda_2)$.
Following the proof for Lemma~\ref{lem:app}, we obtain the following lemma.

\begin{lemma}\label{lem:app2}
Let $\phi_\bfv:\cP(m,L;\bfv,t)\to \bigcup_{r=0}^t \Lambda(r)$ be defined as above.
Then $\phi_\bfv$ is a well-defined injective map.
\end{lemma}

Therefore, Propostion~\ref{prop:Nt-secc} follows from the fact that $\phi_\bfv$ is injective and the formula \eqref{eq:Lambda}.

\section{Verification of Optimality Certificates}\label{app:opt}

\begin{table*}
\begin{center}
\scriptsize
\renewcommand{\arraystretch}{1.6}
\begin{tabular}{|c|l|l|}
\hline
$m$ & $\widetilde{\bfY}^*$ & $\widetilde{\bfX}^*$\\
\hline
$m\equiv 0\pmod4$ &
$\widetilde{Y}^*_{\bfL(m-1)}=\frac{1}{mL}$, &
$\widetilde{X}^*_{\bfL(m)}=1$,~~
$\widetilde{X}^*_{\bfL(m-4i-3)}=\frac{1}{4i+3}\binom{m}{4i+2}L^{4i+2}$ ,
\\
& 
$\widetilde{Y}^*_{\bfL(m-4i-2)}=\widetilde{Y}^*_{\bfL(m-4i-3)}=\frac{1}{4i+4}$ for $0\le i< \floor{\frac m4}$ & 
$\widetilde{X}^*_{\bfL(m-4i-4)}=\frac{1}{4i+4}\left(\binom{m}{4i+3}L^{4i+3}-\frac{1}{4i+3}\binom{m}{4i+2}L^{4i+2}\right)$ for $0\le i<\floor{\frac m4}$ 
\\
\hline
$m\equiv 1\pmod4$ &
$\widetilde{Y}^*_{\bfL(m)}=1$, &
$\widetilde{X}^*_{\bfL(m)}=1$,~~
$\widetilde{X}^*_{\bfL(m-4i-4)}=\frac{1}{4i+4}\binom{m}{4i+3}L^{4i+3}$ ,
\\
& 
$\widetilde{Y}^*_{\bfL(m-4i-3)}=\widetilde{Y}^*_{\bfL(m-4i-4)}=\frac{1}{4i+5}$ for $0\le i< \floor{\frac m4}$ & 
$\widetilde{X}^*_{\bfL(m-4i-5)}=\frac{1}{4i+5}\left(\binom{m}{4i+4}L^{4i+4}-\frac{1}{4i+4}\binom{m}{4i+3}L^{4i+3}\right)$ for $0\le i<\floor{\frac m4}$ 
\\
\hline
$m\equiv 2\pmod4$ &
$\widetilde{Y}^*_{\bfL(m-4i)}=\widetilde{Y}^*_{\bfL(m-4i-1)}=\frac{1}{4i+2}$ for $0\le i\le \floor{\frac m4}$ & 
$\widetilde{X}^*_{\bfL(m-4i-1)}=\frac{1}{4i+1}\binom{m}{4i}L^{4i}$ ,
\\
& 
&
$\widetilde{X}^*_{\bfL(m-4i-2)}=\frac{1}{4i+2}\left(\binom{m}{4i+1}L^{4i+1}-\frac{1}{4i+1}\binom{m}{4i}L^{4i}\right)$ for $0\le i\le\floor{\frac m4}$ 
\\
\hline
$m\equiv 3\pmod4$ &
$\widetilde{Y}^*_{\bfL(m-4i-1)}=\widetilde{Y}^*_{\bfL(m-4i-2)}=\frac{1}{4i+3}$ for $0\le i\le \floor{\frac m4}$ & 
$\widetilde{X}^*_{\bfL(m-4i-2)}=\frac{1}{4i+2}\binom{m}{4i+1}L^{4i+1}$ ,
\\
& 
&
$\widetilde{X}^*_{\bfL(m-4i-3)}=\frac{1}{4i+3}\left(\binom{m}{4i+2}L^{4i+2}-\frac{1}{4i+2}\binom{m}{4i+1}L^{4i+1}\right)$ for $0\le i\le\floor{\frac m4}$ 
\\
\hline

\end{tabular}
\end{center}
Here, we use $\bfL(i)$ to denote the $m$-tuple $[\underbrace{L,L,\ldots,L}_{i\text{ times}},\underbrace{L-1,L-1,\ldots,L-1}_{m-i\text{ times}}]$.
\caption{List of Optimality Certificates for Proposition~\ref{prop:secc1}}
\label{table:secc1}
\end{table*}

\begin{table*}
\begin{center}
\footnotesize
\renewcommand{\arraystretch}{1.6}
\begin{tabular}{|c|l|l|}
\hline
$L-\ws$ & $\widetilde{\bfY}^*$ & $\widetilde{\bfX}^*$\\
\hline
$L-\ws\equiv 0\pmod4$ &
$\widetilde{Y}^*_{[L]}=1$, &
$\widetilde{X}^*_{[L]}=1$,~~
$\widetilde{X}^*_{[L-4i-3]}=\frac{1}{4i+3}\binom{L}{4i+2}$ ,
\\
& 
$\widetilde{Y}^*_{[L-4i-2]}=\widetilde{Y}^*_{[L-4i-3]}=\frac{1}{4i+4}$ for $0\le i< \floor{\frac{L-\ws}4}$ & 
$\widetilde{X}^*_{[L-4i-4]}=\frac{1}{4i+4}\left(\binom{L}{4i+3}-\frac{1}{4i+3}\binom{L}{4i+2}\right)$ for $0\le i<\floor{\frac m4}$ 
\\
\hline
$L-\ws\equiv 1\pmod4$ &
$\widetilde{Y}^*_{[L]}=1$, &
$\widetilde{X}^*_{[L]}=1$,~~
$\widetilde{X}^*_{[L-4i-4]}=\frac{1}{4i+4}\binom{L}{4i+3}$ ,
\\
& 
$\widetilde{Y}^*_{[L-4i-3]}=\widetilde{Y}^*_{[L-4i-4]}=\frac{1}{4i+5}$ for $0\le i< \floor{\frac{L-\ws}4}$ & 
$\widetilde{X}^*_{[L-4i-5]}=\frac{1}{4i+5}\left(\binom{L}{4i+4}-\frac{1}{4i+4}\binom{L}{4i+3}\right)$ for $0\le i<\floor{\frac{L-\ws}4}$ 
\\
\hline
$L-\ws\equiv 2\pmod4$ &
$\widetilde{Y}^*_{[L-4i]}=\widetilde{Y}^*_{[L-4i-1]}=\frac{1}{4i+2}$ for $0\le i\le \floor{\frac{L-\ws}4}$ & 
$\widetilde{X}^*_{[L-4i-1]}=\frac{1}{4i+1}\binom{L}{4i}$ ,
\\
& 
&
$\widetilde{X}^*_{[L-4i-2]}=\frac{1}{4i+2}\left(\binom{L}{4i+1}-\frac{1}{4i+1}\binom{L}{4i}\right)$ for $0\le i\le\floor{\frac{L-\ws}4}$ 
\\
\hline
$L-\ws\equiv 3\pmod4$ &
$\widetilde{Y}^*_{[L-4i-1]}=\widetilde{Y}^*_{[L-4i-2]}=\frac{1}{4i+3}$ for $0\le i\le \floor{\frac{L-\ws}4}$ & 
$\widetilde{X}^*_{[L-4i-2]}=\frac{1}{4i+2}\binom{L}{4i+1}$ ,
\\
& 
&
$\widetilde{X}^*_{[L-4i-3]}=\frac{1}{4i+3}\left(\binom{L}{4i+2}-\frac{1}{4i+2}\binom{L}{4i+1}\right)$ for $0\le i\le\floor{\frac{L-\ws}4}$ 
\\
\hline
\end{tabular}
\end{center}
\caption{List of Optimality Certificates for Proposition~\ref{prop:secc2}}
\label{table:secc2}
\end{table*}

We establish Propositions~\ref{prop:secc1} and \ref{prop:secc2} 
by providing the optimality certificates for the linear program \eqref{gsp:SECC}.
In Tables~\ref{table:secc1} and \ref{table:secc2}, we provide the optimality certificates 
$\left(\widetilde{\bfY}^*,\widetilde{\bfX}^*\right)$ for Propositions~\ref{prop:secc1} and~\ref{prop:secc2}, respectively.

In this appendix, we also provide a detailed verification for the case $w=L-1$, $m\ge L/2$, and $m\equiv 0\pmod{4}$.
We omit the detailed verification for the other cases as the verification process is similar.

For brevity, we adopt the following notations:
\begin{align*}
\bfL(i) &\triangleq [\underbrace{L,\ldots,L}_{i\text{ times}},\underbrace{L-1,\ldots,L-1}_{m-i\text{ times}}],\\
\bfK(i) &\triangleq [\underbrace{L,\ldots,L}_{i\text{ times}},\underbrace{L-1,\ldots,L-1}_{m-i-1\text{ times}}, L-2].
\end{align*}
We also write $\cP_{\rm row}(m,L;\ws)$ and $\cP_{\rm col}(m,L;\ws)$ 
as $\cP_{\rm row}$ and $\cP_{\rm col}$, respectively.

Observe that  $\cP_{\rm row}=\{ \bfL(i): 0\le i\le m\}$ and 
$\cP_{\rm col}=\{ \bfL(i): 0\le i\le m\}\cup \{ \bfK(i): 0\le i\le m-1\}$.
We now state explicitly the entries of $\bfM$. 
For $0\le i,j \le m$, we have that
\[ \bfM_{\bfL(i), \bfL(j)} = 
\begin{cases}
m-i,	& \mbox{if $i=j-1$},\\
1,	& \mbox{if $i=j$},\\
Li,	& \mbox{if $i=j+1$},\\
0, 	& \mbox{otherwise}.
\end{cases}
\]
For $0\le i\le m$ and $0\le j \le m-1$, we have that
\[ \bfM_{\bfL(i), \bfK(j)} = 
\begin{cases}
(L-1)(m-i),	& \mbox{if $i=j$},\\
0, 	& \mbox{otherwise}.
\end{cases}
\]
Next, we state explicitly the entries of $\bfO$.
\begin{align*}
|O_{\bfL(i)}| & = \binom{m}{i}L^{m-i}, &\mbox{for $0\le i\le m$},\\
|O_{\bfK(i)}| & = m\binom{m-1}{i}L^{m-i-1}\binom{L}{2}, &\mbox{for $0\le i\le m-1$}.
\end{align*}

We first verify property (i) for Definition~\ref{def:opt-cert}. In particular, we check that $\bfM^*\widetilde{\bfY}^*\ge \one$, or equivalently, $\sum_{\bfv\in \cP_{\rm col}}\bfM_{\bfL(m-i),\bfv}\widetilde{Y}^*_{\bfv}\ge 1$ for all $0\le i\le m$. Now, since all entries of $Y_\bfv$ are nonnegative, we have that
\[ \sum_{\bfv\in \cP_{\rm col}}\bfM_{\bfL(m-i),\bfv}\widetilde{Y}^*_{\bfv}\ge 
\sum_{j=i-1}^{i+1}\bfM_{\bfL(m-i),\bfL(m-j)}\widetilde{Y}^*_{\bfL(m-j)}.\]

Then we have the following cases.
\begin{itemize}
\item When $i=0$, the quantity on the righthand side is bounded below by 
\[\bfM_{\bfL(m),\bfL(m-1)}\widetilde{Y}^*_{\bfL(m-1)}=mL(1/mL)= 1.\]

\item When $i\equiv 0\pmod 4$, the quantity on the righthand side is bounded below by 
\[\bfM_{\bfL(m-i),\bfL(m-i+1)}\widetilde{Y}^*_{\bfL(m-i+1)}=i(1/i)= 1.\]

\item When $i\equiv 1\pmod 4$, the quantity on the righthand side is bounded below by 
\[\bfM_{\bfL(m-i),\bfL(m-i-1)}\widetilde{Y}^*_{\bfL(m-i-1)}=(m-i)L/(i+3)\ge 1.\]

\item When $i\equiv 2\pmod 4$, the quantity on the righthand side is bounded below by 
\[\bfM_{\bfL(m-i),\bfL(m-i-1)}\widetilde{Y}^*_{\bfL(m-i-1)}=(m-i)L/(i+2)\ge 1.\]

\item When $i\equiv 3\pmod 4$, the quantity on the righthand side is bounded below by
{\small 
\begin{align*}
&\bfM_{\bfL(m-i),\bfL(m-i+1)}\widetilde{Y}^*_{\bfL(m-i+1)}
+\bfM_{\bfL(m-i),\bfL(m-i)}\widetilde{Y}^*_{\bfL(m-i)}\\
&\hspace{47mm}=i/(i+1)+1/(i+1)=1.
\end{align*}
}
\end{itemize}

\vspace{2mm}

Next, we verify property (ii) for Definition~\ref{def:opt-cert}. 
In particular, we check that $\widetilde{\bfX}^*\bfM^*\le \bfO$, or equivalently, 
\begin{align*}
\sum_{\bfv\in \cP_{\rm row}}\widetilde{X}^*_{\bfv}\bfM_{\bfv,\bfL(m-i)}
& \le |O_{\bfL(m-i)}|
& \mbox{ for all $0\le i\le m$},\\ 
\sum_{\bfv\in \cP_{\rm row}}\widetilde{X}^*_{\bfv}\bfM_{\bfv,\bfK(m-i)}
& \le |O_{\bfK(m-i)}|
& \mbox{ for all $1\le i\le m$}.
\end{align*}
Now, let us focus on the column indexed by $\bfL(m-i)$ for $0\le i\le m$.
Recall that $|O_{\bfL(m-i)}|=\binom{m}{i}L^i$.
Since most entries of $\bfM$ are zero, we have that
\[\sum_{\bfv\in \cP_{\rm row}}\widetilde{X}^*_{\bfv}\bfM_{\bfv,\bfL(m-i)}
=\sum_{j=i-1}^{i+1}\widetilde{X}^*_{\bfL(m-j)}\bfM_{\bfL(m-j),\bfL(m-i)}.\]
Then we have the following cases.
\begin{itemize}
\item When $i=0$, the quantity on the righthand side is 
\[\widetilde{X}^*_{\bfL(m)}\bfM_{\bfL(m),\bfL(m)}=1= |O_{\bfL(m)}|.\]

\item When $i=1$, the quantity on the righthand side is 
\[\widetilde{X}^*_{\bfL(m)}\bfM_{\bfL(m),\bfL(m-1)}=1(mL)= |O_{\bfL(m-1)}|.\]

\item When $i\equiv 0\pmod{4}$ and $i\ne 0$, the quantity on the righthand side is 
{\small
\begin{align*}
&\widetilde{X}^*_{\bfL(m-i+1)}\bfM_{\bfL(m-i+1),\bfL(m-i)} + \widetilde{X}^*_{\bfL(m-i)}\bfM_{\bfL(m-i),\bfL(m-i)}\\
& \hspace{10mm} = \left(\frac{1}{i-1}\binom{m}{i-2}L^{i-2}\right)L(m-i+1) \\
& \hspace{13mm}+ \frac1i \left( \binom{m}{i-1}L^{i-1}-\frac{1}{i-1}\binom{m}{i-2}L^{i-2}\right)\\
& \hspace{10mm}\le \left(\frac{m-i+1}{i-1}\binom{m}{i-2} +\frac{1}{i}\binom{m}{i-1} \right)L^{i-1}\\
& \hspace{10mm}\le \binom{m}{i} L^{i}= |O_{\bfL(m-i)}|.
\end{align*}
}

\item When $i\equiv 1\pmod{4}$ and $i\ne 1$, the quantity on the righthand side is 
{\small
\begin{align*}
&\widetilde{X}^*_{\bfL(m-i+1)}\bfM_{\bfL(m-i+1),\bfL(m-i)} \\
& \hspace{2mm} = \frac1{i-1} \left( \binom{m}{i-2}L^{i-2}-\frac{1}{i-2}\binom{m}{i-3}L^{i-3}\right)L(m-i+1)\\
& \hspace{2mm}\le \frac{m-i+1}{i-1}\binom{m}{i-2}L^{i-1}\\
& \le \binom{m}{i} L^{i}= |O_{\bfL(m-i)}|.
\end{align*}
}

\item When $i\equiv 2\pmod{4}$, the quantity on the righthand side is 
\[ \widetilde{X}^*_{\bfL(m-i-1)}\bfM_{\bfL(m-i-1),\bfL(m-i)}=\left(\frac{1}{i}\binom{m}{i}\right)i= |O_{\bfL(m-i)}|.\]

\item When $i\equiv 3\pmod{4}$, the quantity on the righthand side is 
{\footnotesize
\begin{align*}
&\widetilde{X}^*_{\bfL(m-i)}\bfM_{\bfL(m-i),\bfL(m-i)} + \widetilde{X}^*_{\bfL(m-i-1)}\bfM_{\bfL(m-i-1),\bfL(m-i)}\\
& \hspace{1mm} = \frac{1}{i}\binom{m}{i-1}L^{i-1} 
+ \frac{1}{i+1} \left( \binom{m}{i}L^{i}-\frac{1}{i}\binom{m}{i-1}L^{i-1}\right)(i+1)\\
& \hspace{1mm}=\binom{m}{i} L^{i}=|O_{\bfL(m-i)}|.
\end{align*}
}
\end{itemize}

Next, we look at the column indexed by indexed by $\bfK(m-i)$ for $1\le i\le m$.
Recall that $|O_{\bfK(m-i)}| = m\binom{m-1}{i-1}L^{i-1}\binom{L}{2}$.
Since most entries of $\bfM$ are zero, we have that
\begin{align*}
&\sum_{\bfv\in \cP_{\rm row}}\widetilde{X}^*_{\bfv}\bfM_{\bfv,\bfK(m-i)}\\
&\hspace{10mm}=\widetilde{X}^*_{\bfL(m-i)}\bfM_{\bfL(m-i),\bfK(m-i)}=\widetilde{X}^*_{\bfL(m-i)}(L-1)i.
\end{align*}

Since $\widetilde{X}^*_{\bfL(m-i)}$ is nonzero only in certain instances, we consider the following.
\begin{itemize}
\item When $i=1$, the quantity on the righthand side is 
\[L-1< m\binom{L}{2}=|O_{\bfK(m-1)}|.\]

\item When $i\equiv 0\pmod{4}$ and $i\ne 0$, the quantity on the righthand side is 
{
\begin{align*}
&\frac1i \left( \binom{m}{i-1}L^{i-1}-\frac{1}{i-1}\binom{m}{i-2}L^{i-2}\right)(L-1)i\\
& \hspace{10mm}\le \left(\frac{1}{i}\binom{m}{i-1} L^{i-1} \right) (L-1)i\\
& \hspace{10mm}\le m\binom{m-1}{i-1} L^{i-1}\binom{L}{2}= |O_{\bfK(m-i)}|.
\end{align*}
}

\item When $i\equiv 3\pmod{4}$, the quantity on the righthand side is 
{
\begin{align*}
&\left(\frac1i \binom{m}{i-1}L^{i-1}\right)(L-1)i\\
& \hspace{10mm}\le m\binom{m-1}{i-1} L^{i-1}\binom{L}{2}= |O_{\bfK(m-i)}|.
\end{align*}
}\end{itemize}

\vspace{2mm}

Finally, we verify property (iii) for Definition~\ref{def:opt-cert}. Indeed, we have that
{\small
\begin{align*}
& \sum_{\bfu\in \cP_{\rm col}} |O_\bfu| \widetilde{Y}^*_\bfu\\
& = \frac{mL}{mL}+\sum_{i=0}^{\floor{m/4}-1} \frac{1}{4i+4}\left(\binom{m}{4i+2}L^{4i+2}+\binom{m}{4i+3}L^{4i+3}\right)\\
& = 1+\sum_{i=0}^{\floor{m/4}-1} \frac{1}{4i+4}\left(\binom{m}{4i+2}L^{4i+2}+\binom{m}{4i+3}L^{4i+3}\right);
\end{align*}
}

\noindent and 
{\footnotesize
\begin{align*}
&\sum_{\bfu\in \cP_{\rm row}} \widetilde{X}^*_\bfu\\
& = 1+\sum_{i=0}^{\floor{m/4}-1} \frac{1}{4i+3}\binom{m}{4i+2}L^{4i+2}+\\
&\hspace{20mm}\frac{1}{4i+4}\left(\binom{m}{4i+3}L^{4i+3}-\frac{1}{4i+3}\binom{m}{4i+2}L^{4i+2}\right)\\
& = 1+\sum_{i=0}^{\floor{m/4}-1} \frac{1}{4i+4}\left(\binom{m}{4i+2}L^{4i+2}+\binom{m}{4i+3}L^{4i+3}\right).
\end{align*}
}

\noindent Therefore, equality holds and we have verified the conditions for $(\widetilde{\bfY}^*,\widetilde{\bfX}^*)$ to be an optimality certificate.

\bibliographystyle{IEEEtran}
\bibliography{abrv,conf_abrv,mybibfile}

\begin{thebibliography}{10}
\providecommand{\url}[1]{#1}
\csname url@samestyle\endcsname
\providecommand{\newblock}{\relax}
\providecommand{\bibinfo}[2]{#2}
\providecommand{\BIBentrySTDinterwordspacing}{\spaceskip=0pt\relax}
\providecommand{\BIBentryALTinterwordstretchfactor}{4}
\providecommand{\BIBentryALTinterwordspacing}{\spaceskip=\fontdimen2\font plus
\BIBentryALTinterwordstretchfactor\fontdimen3\font minus
  \fontdimen4\font\relax}
\providecommand{\BIBforeignlanguage}[2]{{%
\expandafter\ifx\csname l@#1\endcsname\relax
\typeout{** WARNING: IEEEtran.bst: No hyphenation pattern has been}%
\typeout{** loaded for the language `#1'. Using the pattern for}%
\typeout{** the default language instead.}%
\else
\language=\csname l@#1\endcsname
\fi
#2}}
\providecommand{\BIBdecl}{\relax}
\BIBdecl

\bibitem{Tandon16_CSCC_TIT}
A.~Tandon, M.~Motani, and L.~R. Varshney, ``Subblock-constrained codes for
  real-time simultaneous energy and information transfer,'' \emph{{IEEE} Trans.
  Inform. Theory}, vol.~62, no.~7, pp. 4212--4227, Jul. 2016.

\bibitem{Zhao16}
S.~Zhao, ``A serial concatenation-based coding scheme for dimmable visible
  light communication systems,'' \emph{{IEEE} Commun. Lett.}, vol.~20, no.~10,
  pp. 1951--1954, Oct. 2016.

\bibitem{Chee14_MCWC}
Y.~M. Chee, Z.~Cherif, J.-L. Danger, S.~Guilley, H.~M. Kiah, J.-L. Kim,
  P.~Sole, and X.~Zhang, ``Multiply constant-weight codes and the reliability
  of loop physically unclonable functions,'' \emph{{IEEE} Trans. Inform.
  Theory}, vol.~60, no.~11, pp. 7026--7034, Nov. 2014.

\bibitem{Chee13}
Y.~M. Chee, H.~M. Kiah, and P.~Purkayastha, ``Matrix codes and multitone
  frequency shift keying for power line communications,'' in \emph{Proc. 2013
  IEEE Int. Symp. Inf. Theory}, Jul. 2013, pp. 2870--2874.

\bibitem{Tandon15_ISIT}
A.~Tandon, M.~Motani, and L.~R. Varshney, ``Real-time simultaneous energy and
  information transfer,'' in \emph{Proc. 2015 IEEE Int. Symp. Inf. Theory},
  Jun. 2015, pp. 1124--1128.

\bibitem{TandonKM18_IT}
A.~Tandon, H.~M. Kiah, and M.~Motani, ``Bounds on the size and asymptotic rate
  of subblock-constrained codes,'' \emph{{IEEE} Trans. Inform. Theory},
  vol.~64, no.~10, pp. 6604--6619, Oct. 2018.

\bibitem{Fazeli15}
A.~Fazeli, A.~Vardy, and E.~Yaakobi, ``Generalized sphere packing bound,''
  \emph{{IEEE} Trans. Inform. Theory}, vol.~61, no.~5, pp. 2313--2334, May
  2015.

\bibitem{Kulkarni2013}
A.~A. Kulkarni and N.~Kiyavash, ``Nonasymptotic upper bounds for deletion
  correcting codes,'' \emph{{IEEE} Trans. Inform. Theory}, vol.~59, no.~8, pp.
  5115--5130, 2013.

\bibitem{CK16}
D.~Cullina and N.~Kiyavash, ``Generalized sphere-packing bounds on the size of
  codes for combinatorial channels,'' \emph{IEEE Trans. Inf. Theory}, vol.~62,
  no.~8, pp. 4454--4465, 2016.

\bibitem{Freiman64A}
C.~Freiman, ``Upper bounds for fixed-weight codes of specified minimum distance
  {(Corresp.)},'' \emph{{IEEE} Trans. Inform. Theory}, vol. IT-10, no.~3, pp.
  246--248, 1964.

\bibitem{Berger67}
E.~Berger, ``Some additional upper bounds for fixed-weight codes of specified
  minimum distance,'' \emph{{IEEE} Trans. Inform. Theory}, vol. IT-13, no.~2,
  pp. 307--308, 1967.

\bibitem{Margot:2003}
F.~Margot, ``Exploiting orbits in symmetric ilp,'' \emph{Mathematical
  Programming}, vol.~98, no. 1-3, pp. 3--21, 2003.

\bibitem{Bodi:2009}
R.~B{\"o}di and K.~Herr, ``Symmetries in linear and integer programs,''
  \emph{arXiv preprint arXiv:0908.3329}, 2009.

\bibitem{Chvatal1983}
V.~Chvatal, \emph{Linear programming}.\hskip 1em plus 0.5em minus 0.4em\relax
  Macmillan, 1983.

\end{thebibliography}
\end{document}